\documentclass{article}
\usepackage{graphicx} 
\usepackage{amsmath}
\usepackage{amsthm}
\usepackage{amssymb}
\usepackage{authblk}
\usepackage{url}
\usepackage{comment}
\usepackage{mathtools}

\DeclarePairedDelimiter\floor{\lfloor}{\rfloor}

\newtheorem{question}{Question}
\newtheorem{definition}{Definition}
\newtheorem{proposition}{Proposition}

\newtheorem{theorem}{Theorem}
\newtheorem{corollary}{Corollary}
\newtheorem{lemma}{Lemma}
\newtheorem{example}{Example}

\usepackage{tikz}
\usetikzlibrary{positioning,chains,fit,shapes,calc}

\definecolor{myblue}{RGB}{80,80,160}
\definecolor{mygreen}{RGB}{80,160,80}

\title{The Cost of Secure Restaking vs. Proof-of-Stake\thanks{We thank Ed Felten, Manvir Schneider, Paolo Penna, Pranay Anchuri, Christoph Schlegel and Bruno Mazorra for useful discussions. All errors are our own.}}
\author[1]{Akaki Mamageishvili}
\author[2]{Benny Sudakov}
\affil[1]{Offchain Labs}
\affil[2]{ETH Zurich}

\date{March 2026}

\begin{document}

\maketitle

\begin{abstract}

We compare the total capital efficiency of secure restaking and Proof-of-Stake (PoS) protocols. First, we consider the sufficient condition for the restaking graph to be secure. The condition implies that it is always possible
to transform such a restaking graph into separate secure PoS protocols.
Next, we derive two main results: upper and lower bounds on the required extra stakes to add to the validators of the secure restaking graph to be able to transform it into secure PoS protocols. In particular, we show that the restaking savings compared to PoS protocols can be very large and can asymptotically grow as a square root of the number of validators. We also study a complementary question of aggregating secure PoS protocols into a secure restaking graph and provide matching lower and upper bounds on the PoS savings.
\end{abstract}

\section{Introduction}

The cryptoeconomic security of blockchains has evolved from energy-intensive Proof-of-Work to more efficient Proof-of-Stake (PoS), and most recently to restaking protocols~\footnote{Total value locked in all restaking protocols amounts to more than $\$20B$ -- \url{https://defillama.com/protocols/restaking}, accessed on January, 2026.}. Restaking allows validators to reuse their staked capital across multiple services or projects, creating a bipartite graph between validators and services. This reuse promises significant capital efficiency gains over running independent PoS protocols for each service -- but it also introduces complex interdependencies and potential security risks.

It can be argued that restaking gives a boost to newer projects in attracting capital for their {\it cryptoeconomic security}. Cryptoeconomic security refers to a security derived from staked parties being slashed their stakes in case of misbehavior. Restaking seems to be attractive for the validators of the original protocol as well. In particular, restaking allows them to reuse idle computational power and storage allocation for additional rewards from other projects. The original protocol, for which the validators are staked, may have requirements that do not exhaust all the resources of the validators. Setting low requirements is usually done for the purpose of decentralization\footnote{Check out Ethereum network plans to resolve this asymmetry between validator resources using rainbow staking~\url{https://ethresear.ch/t/unbundling-staking-towards-rainbow-staking/18683}.}. Another potential reason for such availability might be economies of scale for establishing further validation services.

The simplest way to verify the security of the restaking uses a trivial cost-benefit analysis. In this analysis, the total amount of lost stakes of attackers is compared to the total value obtained from the attack of a set of validators~\footnote{~\cite{ces} discusses ways to estimate the value gained from attacking a protocol and cryptoeconomic security in more detail.}. 
The same cost-benefit analysis can be used to check the security of any PoS protocol, in which a single project/service is secured by multiple validators. In fact, the cost-benefit analysis for restaking is a direct generalization of the cost-benefit analysis for PoS protocols. This allows a comparison of restaking and multiple PoS protocols, which we develop in this paper.

Conceptually, the PoS protocol is simpler and a more natural concept than restaking. Hence, we are interested in checking comparative benefits of restaking protocols. 
One particular attractive point of restaking protocols can be the lower total staking requirements than that of PoS protocols, since there are opportunity costs associated with locking up stake for security reasons. A key open question is whether restaking truly reduces total staking requirements compared to equivalent secure PoS protocols, and if so, by how much.

Most PoS protocols pay for the opportunity cost of locking up a capital by rewarding stakers with newly issued tokens, which leads to inflation. Hence, minimizing the total stake amount while maintaining security can be a reasonable goal in both restaking or PoS protocols. 
Starting with a secure restaking protocol and initial stake endowments for all validators, we measure how much extra stake would be required if validators decided to form separate Proof-of-Stake protocols with their stake amounts, dividing them across the same set of projects they secure in the original restaking graph.

In the following, we describe an example that gives intuition of the secure restaking graph and its transformation to secure PoS protocols, see the Figure 1.

\begin{example}
    There are two services, with valuations $1$ and $3$ and security parameters $\alpha$, which is any number in the interval $(0,1)$ and $\frac{1}{2}$. There are two validators, with stakes $2$ and $2$. 
\end{example}
The security parameter of a service establishes what share of the total stake securing the service is needed to successfully attack this service. Typical values of security parameters are from the set $\{\frac{1}{2}, \frac{1}{3}, \frac{1}{4}\}$~\footnote{These values are mostly coming from the {\it Byzantine Fault Tolerant} protocols for achieving the consensus over the state. However, the values can be lower too.}.

Let us first check that the restaking graph is secure. It is easy to see: 

\begin{itemize}
    \item the first validator alone can attack only the first service, and it is not profitable;
    \item  the second validator alone can not attack any service: it is not securing the first service, while the second service has large enough security parameter.  
    \item finally, both validators can attack both services but they lose the same amount of stake as they gain the value.
\end{itemize} 

Next, we check if the validators can distribute their stakes so that both services are secured in the corresponding PoS protocols. First, note that it is the first validator that divides its stake between two services. The first service needs at least the stake amount $1$ from the first validator. The second service is not secure in the PoS protocol with stakes $1$ and $2$. In fact, it can be easily seen that either the first validator needs to have at least a matching stake $2$, or the second validator needs to have at least a stake $3$. In either case, at least an extra stake $1$ is required to secure both services in PoS protocols.

\begin{figure}
\centering
\begin{tikzpicture}[thick,
  every node/.style={draw,circle},
  snode/.style={fill=myblue},
  vnode/.style={fill=mygreen},
  every fit/.style={ellipse,draw,inner sep=-2pt,text width=2.5cm},
  ->,shorten >= 3pt,shorten <= 3pt
]
\begin{scope}[start chain=going below,node distance=7mm]
  \node[snode,on chain] (s1) [label=left: \text{1}] {};
  \node[snode,on chain] (s2) [label=left: \text{3}] {};
\end{scope}

\begin{scope}[xshift=4cm,yshift=-0.0cm,start chain=going below,node distance=7mm]
  \node[vnode,on chain] (v1) [label=right: \text{2}] {};
  \node[vnode,on chain] (v2) [label=right: \text{2}] {};
\end{scope}

\node [myblue,fit=(s1) (s2),label=above:$S$] {};
\node [mygreen,fit=(v1) (v2),label=above:$V$] {};

\draw (v1) -- (s1);
\draw (v1) -- (s2);
\draw (v2) -- (s2);
\end{tikzpicture}
\label{fig:example}
\caption{Example of a secure restaking graph.}
\end{figure}

We only allow all validator stakes to (weakly) increase. This captures a real-life consideration in which validators' stakes can not be taken away, but their stakes can be increased through financing when needed. In a sense, we look into over-collateralization of the whole network needed to allow separate secure PoS protocols. Similarly,~\cite{robust_restaking} investigates the over-collateralization of individual validators (instead of the full set of validators as we do in this paper) that allows robustness of the restaking network security. Most of our results can be extended to the model with individual over-collaterization, obtaining the same bounds. 

If both increasing and decreasing of validator stakes were allowed, it would be trivial to construct a restaking graph with minimum total stake. This amount is equal to the sum of all project values, and secure PoS protocols for all projects can be constructed by dividing the stakes of the validators. See Proposition~\ref{minimum_requirement} and discussion before it.

\paragraph*{Our Contributions}
\begin{enumerate}
    \item We check the sufficient condition for a secure restaking graph, specified by the EigenLayer project~\cite{eigenlayer}. This sufficient condition is the only currently available (generic) condition for secure restaking graph, that can be checked in polynomial time in the size of the restaking graph. Restaking graphs that satisfy the sufficient condition do not save on the total required stake. In particular, it is possible to divide the initial stakes of the validators across projects they validate, to obtain secure PoS protocols.

    \item To further extend the analysis, we define {\it restaking savings} of a fixed secure restaking graph. It is equal to the extra stake required, relative to the original total available stake in the graph, to be able to obtain secure PoS protocols through the division of stake described above.

    \item To upper bound restaking savings, we use simple secure PoS constructions from the original restaking graph. In particular, we upper bound the value of restaking savings with the highest degree among the projects. Another upper bound is the minimum incidence of any validator in the cover of a validator set by project neighborhoods. The third upper bound is equal to second largest multiplicative inverse of the security parameter, which is a small constant value in many practical cases. The last upper bound is equal to the square root of the number of validators. Next, we give a construction to lower bound this value. The lower bound example asymptotically matches all upper bounds obtained.     
    \item We study a complementary question of aggregating given secure PoS protocols into a secure restaking graph. When aggregating, each validator adds stakes across projects that they secure in separate PoS protocols, resulting in a restaking graph. 
    Similarly to the restaking savings, we define {\it PoS savings} as the total additional stake required to add validators in the aggregate restaking graph to make it secure.
    We first construct an example showing that the resulting restaking graph is not secure. In fact, at least an additional stake equal to the original total stake times the number of projects/validators minus one is required to make it secure, thus showing a lower bound of number of projects/validators minus one on PoS savings. We also provide an upper bound on PoS savings, which exactly matches a lower bound. 
\end{enumerate}

\subsection*{Related Work}
\cite{bitcoin} introduced the first cryptocurrency Bitcoin and the concept of PoW, in which the random sampling of the next block proposer is done proportionally to a hashrate of a miner, by solving a computational puzzle.~\cite{PoS} provides a detailed discussion and economic comparison on PoW and PoS concepts.~\cite{ces} discusses ways to estimate the value gained from attacking a protocol and cryptoeconomic security in more detail.
The idea of restaking was originally proposed in the whitepaper of the Eigenlayer project,~\cite{eigenlayer}. A similar concept exists in traditional finance literature under the name of {\it rehypothecation},~\cite{monnet}.
The interdependencies between protocols and validators introduced by restaking create new risks, which have been modeled and analyzed in the original Eigenlayer whitepaper, as well as in a follow up paper~\cite{robust_restaking}, in which the authors study robust restaking networks, resistant to a failure of a fraction of validating stakes and cascading effects. The authors consider over-collaterization factor of individual validators. They investigate out how large should this uniform over-collaterization factor be so that the originally secure restaking graph stays robust under these crashes. Our over-collaterization over the total stake amount can also be reduced to individual over-collaterization factors, as we argue after the Theorem~\ref{lower_bound}.

~\cite{how_much} studies a setting in which the value of attacked set of services is not the sum of values in the set, but a concave function.~\cite{elastic_restaking} studies elastic restaking services, in which validators allocate stakes to services that may make up more than their initial stake. This allows for a more expressive model than the original restaking model where the whole stake of each validator is locked up for all services. The work also considers a typical distributed computing setting in which only a fraction of validators is Byzantine.~\cite{economic_security} studies risks and benefits of combining different restaking services.~\cite{limits} studies economic security of different protocols as a function of streaming rewards.

\section{Notation and Model}
Our notation follows closely~\cite{robust_restaking}.
A restaking graph $G=(S, \pi, \alpha, V,\sigma, E)$ consists of the following:

\begin{itemize}
    \item $S$ denotes the set of $m$ services $\{1,...,m\}$ (we will refer projects as services from now on throughout the paper),
    \item $\pi \in \mathbb{R}^{m}_{+}$ denotes values of services,
    \item $\alpha \in \mathbb{R}^m_{+}$ denotes security parameters of services, where $\alpha_s\in [0,1]$,
    \item $V$ denotes the set of $n$ validators $\{1,...,n\}$,
    \item $\sigma\in \mathbb{R}^n_+$ denotes stakes of validators,
    \item The set of edges $E$ connecting $S$ nodes (services) to $V$ nodes (validators).
\end{itemize}

An edge $(s,v)\in E$ indicates that a service $s$ is secured by a validator $v$. Hence, a restaking graph is represented as a bipartite graph with two parts made up of services and validators. 


Let $N_G(s)$ denote the neighborhood of the service $s\in S$ in the restaking graph $G$, i.e. $N_G(s):=\{v \in V: (s,v)\in E\}$, is the set of validators staked on, and, therefore, securing a service $s$.
Similarly, $N_G(v)$ denotes the neighborhood of the validator $v\in V$, i.e., $N_G(v):=\{s \in S: (s,v)\in E\}$, is the set of services validator $v$ is staked on. 

A service $s$ has a security parameter $\alpha_s$ if only subsets of validators staked on this server that have at least $\alpha_s$ fraction of the total stake of validators staked on $s$ can attack it. More formally:

\begin{definition}
   A subset of validators $W\subseteq V$ can attack the service $s$ if and only if $$\frac{\sum_{k\in W\cap N_G(s)}\sigma_k}{\sum_{{v \in N_G(s)}}\sigma_v}>\alpha_s.$$
\end{definition}

For each subset of validators $A$, we can define the maximal set of services that they can collectively attack. Let this set be denoted by $M(A)$. That is $M(A)=\{s\in S: A~ \text{can attack }s\}$. The attack carried out by the validators in $A$ is profitable when $\sum_{s \in M(A)}\pi_s>\sum_{v \in A}\sigma_v$, that is, validators lose less total stake than the total value they obtain from the attack. 

\begin{definition}\label{def:secure_restaking}
    A restaking graph $G$ is secure if there exists no subset of validators $A$ that can profitably attack its corresponding set of services $M(A)$.
\end{definition}

We are interested in the following question:

\begin{question}
Is it possible to divide the stakes of all validators between all services they secure in the graph $G$ so that all services are secure in their corresponding PoS protocols?    
\end{question}
    
 A division of stake $\sigma_v$ of validator $v$ across services it secures in $G$ induces a vector $c^v\in \mathbb{R}^{|N_G(v)|}_+$, such that, $\sum_{s\in N_G(v)}c^{v}_{s} = \sigma_v$. 
We say that division of stake is secure if in the corresponding Proof-of-Stake protocols, that represent star graphs centered at the service, we have guaranteed security. That is, for each service $s$ with security parameter $\alpha_s$, a vector of stakes $\{k\in N_G(s): c^k_{s}\}$ is secure, which means that there is no profitable attack by validators having stakes $\{k\in N_G(s): c^k_{s}\}$. Let $\mathcal{W}(G)$ denote all initial stake vectors $\sigma$ for which such division is possible. To simplify notation, here we implicitly assume that the stake vector is not part of a restaking graph $G$, and only keep its combinatorial structure together with the $\pi$ and $\alpha$ vectors.

We consider stake vectors $\sigma'$ that (weakly) dominate stake vector $\sigma$, that is, $\sigma'_v\geq \sigma_v$ for any $v\in V$. The set of all stake vectors that dominate $\sigma$ is denoted by $D(\sigma)$. Let $T(\sigma):=\sum_{v\in V}\sigma_v$.  Then, the total extra value to reach from the stake vector $\sigma$ to the stake vector $\sigma'$ is $T(\sigma')-T(\sigma)$. The rationale behind considering (weakly) dominant vectors of stakes is that we assume that initial stakes cannot be taken away from the validators, but they can be increased if needed.

The Eigenlayer project,~\cite{eigenlayer}, specifies a sufficient condition when a restaking graph is secure: for any $v\in V$, the following inequality holds:

\begin{equation}\label{sufficient_eigenlayer}
    \sigma_v\geq \sum_{s \in N_G(v)}\frac{\sigma_v}{\sum_{k\in N_G(s)}\sigma_k}\frac{\pi_s}{\alpha_s}.
\end{equation} 

This condition allows to check whether a given restaking graph is secure in polynomial time in the input graph size. In fact, the time is even linear in the number of edges $|E|$~\footnote{The general decision problem whether a given restaking graph is secure is difficult, see~\cite{robust_restaking}.}. 

\section{Results}
In this section, we first check the sufficient condition~\eqref{sufficient_eigenlayer} and its implications. The implications motivate to define restaking and Proof-of-Stake savings. Last, we obtain lower- and upper-bound results for them both types of savings. 
\subsection{Restaking Savings}
If the sufficient condition on validator stakes~\eqref{sufficient_eigenlayer} is satisfied, there is a division of stakes such that all services are secure in their corresponding PoS protocols. More formally, we obtain the following result: 

\begin{proposition}\label{eigenlayer_not_efficient}
When the stake vector $\sigma$ satisfies~\eqref{sufficient_eigenlayer}, then $\sigma\in \mathcal{W}(G)$.
\end{proposition} 

\begin{proof}
    Consider the following division of stakes: the validator $v$ allocates at least $$\frac{\sigma_v}{\sum_{k\in N_G(s)}\sigma_k}\frac{\pi_s}{\alpha_s}$$ stake to a service $s$ it secures in $G$, in the PoS protocol. Such allocation is possible because of~\eqref{sufficient_eigenlayer}. Each service $s$ is assigned a total stake amount of at least $\frac{\pi_s}{\alpha_s}$, since $$\sum_{v\in N_G(s)}\frac{\sigma_v}{\sum_{k\in N_G(s)}\sigma_k}\frac{\pi_s}{\alpha_s}=\frac{\pi_s}{\alpha_s},$$ by bringing the outer summation operator into the nominator. With $\frac{\pi_s}{\alpha_s}$ total stake amount in PoS protocol for service $s$, there is no profitable attack, since for the successful attack the attackers should control at least $\frac{\pi_s}{\alpha_s}\alpha_s = \pi_s$ stake, hence, it cannot be profitable.
\end{proof}

Note that the value $\frac{\pi_s}{\alpha_s}$ is the minimum required total stake in a PoS protocol for a service $s$ with value $\pi_s$ and security parameter $\alpha_s$ that guarantees that there is no profitable attack, independent of how this total stake is divided between validators. 
For any smaller amount $t<\frac{\pi_s}{\alpha_s}$, there exists a distribution of the total stake $t$ among any number of validators such that there is a profitable attack. 

On the other hand, if we consider suitable division of stakes between validators, we obtain a much weaker requirement to make secure PoS protocols. Namely, we obtain that there is a minimum total stake requirement to ensure that there is a division in secure PoS protocols.

However, such a distribution does not take into account that validators already have stakes, which cannot be taken away from them.

\begin{proposition}\label{minimum_requirement}
 The minimum stake amount to guarantee secure PoS protocols is equal to $\sum_{s\in S}\pi_s$.    
\end{proposition}

\begin{proof}
     The total stake amount $\sum_{s\in S}\pi_s$ can be distributed arbitrarily among validators, with the only guarantee that each service can get its own value staked by one of the validators securing it. 
\end{proof}

Another consideration is that the restaking graph with the minimum total stake in $\mathcal{W}(G)$ is not necessarily secure, while we will focus on secure restaking graphs from now on.

We are ready to define the main metric through which we measure capital efficiency of a secure restaking graph, {\it restaking savings}:

\begin{definition}
    For a given secure restaking graph $G$, restaking savings, $RS(G)$, denotes a minimum extra total value to add to the stakes in $\sigma$ vector, in relative terms, so that secure subdivision of the new stake vector is possible. 
    That is, $$RS(G)=\min_{\sigma'\in D(\sigma)\cap \mathcal{W}(G)}\frac{T(\sigma')-T(\sigma)}{T(\sigma)}.$$  
\end{definition}

Let $d_G(u)$ denote the degree of vertex $u\in S\cup V$ in graph $G$. Note that $d_G(u)=|N_G(u)|$. Consider a subset of services whose neighborhoods cover the entire set of validators $V$: 
$$\mathcal{R}:=\{R: R\subseteq S \& \cup_{s\in R}{N_G(s)}=V\}.$$ For each of these covers, calculate the incidences of all validators in it, and take the largest incidence number. Let $K$ be minimum such number over all elements of $\mathcal{R}$. Formally, $$K=\min_{R\in \mathcal{R}}\{k | k=\max_{v\in V}\sum_{s\in R}\mathbb{I}(v\in N_G(s))\},$$ where indicator function $\mathbb{I}(P)=1$ whenever the statement $P$ is true.

We prove several upper bounds on the restaking savings, as a function of different parameters of the graph $G$. Note that since the original graph is secure, we have the following inequality: $$T(\sigma)=\sum_{v \in V} \sigma_v \geq \sum_{s \in S} \pi_s=T(\pi).$$ Otherwise, there is a profitable attack that involves all validators $V$.

\begin{proposition}\label{simple_upper_bounds}
    $RS(G)$ is upper bounded by the following values: $\max_{s\in S}(d_G(s))$ and $K$.
\end{proposition}

\begin{proof}
    First, we list three types of secure PoS constructions for a service $s\in S$. 

    \begin{enumerate}
        \item the vector of stakes in a corresponding PoS protocol $(c_1,...,c_{d_G(s)})$ is equal to $(0,...,0,\pi_s,0,...0)$. This construction is secure for any $\alpha_s$ since the only validator that can attack the service is losing at least the value it is gaining.
        \item the vector of stakes $(c_1,...,c_{d_G(s)})$ is equal to $(0,...,0,\pi_s+\sigma_v,0,...0)$ for any $v$. This construction is secure for a similar reason as above.
        \item the vector of stakes $(c_1,...,c_{d_G(s)})$ is equal to $(\sigma_{v_1},...,\sigma_{v_{d_G(s)}})$, where $N_G(s)=\{v_1,...,v_{d_G(s)}\}$. This construction is secure for any $\alpha_s$ since otherwise the original restaking graph $G$ would not be secure: validators that secure the service $s$ have a profitable attack on the service $s$.
    \end{enumerate}

     We can add the value $\pi_s$ to all stakes in $N_G(s)$, i.e., add it $d_G(s)$ times, and use subdivision for each validator $v$ that allocates $\sigma_v+\pi_s$ once to some service $s$ in its neighborhood $N_G(v)$ and $\pi_k$ to all other corresponding services $k\in N_G(v)\setminus \{s\}$. This way, all services are secure. Also all original stake sizes are used, so that the resulting stake vector is in $D(\sigma)$. At the same time,  extra value added is upper bounded by $$\sum_{s\in S}d_G(s)\pi_s\leq \max_{s\in S}(d_G(s))\sum_{s\in S}\pi_s\leq \max_{s\in S}(d_G(s))\sum_{v\in V}\sigma_v.$$

    Next, we show $RS(G)\leq K$. Consider a cover of vertices in $V$ by neighborhoods of services $s\in R\subseteq S$ that has the minimum maximum incidence number $K$. For each neighborhood $N_G(s)$, where $s\in R$, use construction 3: original $(\sigma_{v_1},...,\sigma_{v_{d_G(s)}})$ stakes from the restaking graph to secure service $s$ in the PoS protocol. To be able to do that, we need to add at most $K-1$ times all $\sigma_v$ stakes, for any $v\in V$. 
    In this way, we utilize all $\sigma_v$ stakes and secure services in the cover set $R$. For service $s$ in $S\setminus R$, we add $\pi_s$ stake once to any of its neighbors in $N_G(s)$ and use construction $(0,...,0,\pi_s,0,...,0)$ to secure service $s$ in a corresponding PoS protocol. This makes sure we have secured all services and utilized all stakes of all validators. Then, $RS(G)\leq (K-1)+1=K$.
\end{proof}

Note that $RS(G)\leq K$ in particular implies that $RS(G)\leq \max_{v\in V}(d_G(v))$, as $K$ is upper bounded by $\max_{v\in V}(d_G(v))$. 
There are examples where $K=\max_{s\in S}(d_G(s))$, implying that the upper bound can be as high as $n$. However, next, we show an upper bound as a function of the number of validators, which is asymptotically much lower than $n$.

\begin{theorem}\label{non_trivial_upper_bound}
For any secure restaking graph $G$, the following inequality holds $RS(G)\leq 2\sqrt{n}-1.$
\end{theorem}

\begin{proof}
The proof relies on a greedy covering of high-degree services. The proof combines constructions from the proof of Proposition~\ref{simple_upper_bounds}.
    Initialize the set $S_L:=\emptyset$ and $V_L:=\emptyset$. The following procedure is repeated until it is no longer possible: At each step, find a service $s$ in $S\setminus S_L$ such that its degree in the set $V\setminus V_L$ is at least $\sqrt{n}$. Update the set $S_L:=S_L\cup s$ and $V_L:=V_L\cup N_G(s)$. Note that after at most $\sqrt{n}$ steps, it is not possible to find a new service $s$ and, hence, the procedure stops. That is, for the resulting set of services $S_L$, we have $|S_L|\leq \sqrt{n}$. This, in particular, implies that the maximum outdegree of a validator $v$ in $V_L$ to $S_L$, denoted by $d_{S_L}(v)$, is upper bounded by $\sqrt{n}$. 
    
    Add to each validator $v\in V_L$ a stake amount $(d_{S_L}(v)-1)\sigma_v$. This allows to secure any service $s\in S_L$ by using a construction $(\sigma_{v_1},...,\sigma_{v_{d_G(s)}})$ from the proof of Proposition~\ref{simple_upper_bounds}. Together with securing all services in $S_L$, we also utilize validator stakes in $V_L$, and this is done by adding at most $(\sqrt{n}-1) T(\sigma)$ extra stake.

    We can secure any service $s\in S\setminus S_L$ by adding stakes $\pi_s$ for to all its neighbors in $V\setminus V_L$, that is, by adding $d_{V\setminus V_L}(s)\pi_s$ extra stakes. By the construction of the set $S_L$ we know that $d_{V\setminus V_L}(s)$ is upper bounded by $\sqrt{n}$ for any $s\in S\setminus S_L$. We can, at the same time, utilize the remaining stakes in $V\setminus V_L$, by using a construction of type $(0,..,0,\pi_s+\sigma_v,..0)$ for $v\in V\setminus V_L$. The last step is possible, because for each validator in $V\setminus V_L$, there must be at least one service in $S\setminus S_L$ it secures. The total extra stake is upper bounded by $\sqrt{n}T(\pi_{S\setminus S_L})\leq \sqrt{n}T(\sigma)$, where $T(\pi_{S\setminus S_L})$ denotes the sum $\sum_{s\in S\setminus S_L}\pi_s$. By summing up two values of extra stakes used to secure $S_L$ and $S\setminus S_L$ sets, and also utilizing all stakes in $V$, we get $$RS(G)\leq \frac{(\sqrt{n}-1)T(\sigma)+\sqrt{n}T(\sigma)}{T(\sigma)}=2\sqrt{n}-1,$$ a required bound of the Theorem claim.
\end{proof}

Note that in the proofs of Proposition~\ref{simple_upper_bounds} and Theorem~\ref{non_trivial_upper_bound}, we did not use any property of $\alpha$ values. Both upper bounds use combinatorial features of the underlying restaking graph $G$ and hold for any values of security parameters $\alpha$. Next, we show an upper bound that is a function of the security parameters $\alpha_s$. We show the following upper bound on $RS(G)$: 

\begin{proposition}\label{RS_upper_bound_alpha}
Assume that services are labeled in the way that the security parameters are sorted in increasing order $\alpha_1\leq \alpha_2\leq ... \leq \alpha_m$. Then, $RS(G)\leq \frac{1}{\alpha_2}$.
\end{proposition}

\begin{proof}
        For the service $s=1$, with the security parameter $\alpha_1$, allocate the original stakes of the restaking graph in the PoS protocol: $(c_1,...,c_{d_G(1)}) = (\sigma_{v_1},...,\sigma_{v_{d_G(1)}})$, where $N_G(1) = \{v_1,...,v_{d_G(1)}\}$ is the set of validators staked on service $s=1$. This PoS protocol is secure by the assumption tjat tje original restaking graph is secure. Repeat the following procedure for all services $s>1$:
        
        \begin{enumerate}
            \item add stake $t_s=\frac{\pi_s}{\alpha_s}$ to any validator $v\in N_G(s)$,
            \item use construction $(\sigma_{v_1},...,\sigma_{v_{d_G(s)}})$ to secure service $s$ in PoS protocol, where $N_G(s) = \{v_1,...,v_{d_G(s)}\}$ is the set of validators staked on service $s$,
            \item set all stake values in the vector $(\sigma_{v_1},...,\sigma_{v_{d_G(s)}})$ to $0$.
        \end{enumerate} 
        
        Each PoS derived at step 2. is secure, as the total stake allocated for a service $s>1$ is at least $\frac{\pi_s}{\alpha_s}$. This is an observation made in the proof of Proposition~\ref{eigenlayer_not_efficient}.
        
        In this way, we add $\sum_{s\in S, s>1}\frac{\pi_s}{\alpha_s}$ total extra stake to the original stakes of validators. This implies the following chain of inequalities: $$RS(G)\leq \frac{\sum_{s\in S, s>1}\frac{\pi_s}{\alpha_s}}{T(\sigma)}\leq \frac{1}{\alpha_2}\frac{\sum_{s\in S,s>1}\pi_s}{T(\pi)}\leq \frac{1}{\alpha_2}, $$
        finishing the proof of the proposition claim.
\end{proof}

Next, we construct an example of a secure restaking graph deriving a lower bound on the restaking savings that is linear in the number of services and asymptotically matches to all upper bounds derived so far.

\begin{theorem}\label{lower_bound}
For any $m\in \mathbb{N}$, there are instances of a secure restaking graph $G$ in which the restaking savings $RS(G)\in \Omega(m)$. 
\end{theorem}

\begin{proof}
    Assume the number of validators is $n=m^2+1$ for some $m\in \mathbb{N}$. The edge set $E$ consists of $m$ edges $(s,m^2+1)$, for any $1\leq s \leq m$ and $m^2$ edges of type $(s,(s-1)m+j)$, for any $1\leq s \leq m$, $1 \leq j \leq m$. 

    The value of service $s$ is defined as $\pi_s=2$ for any $1\leq s\leq m$. Any validator with index less than $v<m^2+1$ has stake $\sigma_v=\frac{1}{m}$. The last validator has stake $\sigma_{m^2+1}=2m$. Security parameter of service $s$ is defined as $\alpha_s=\frac{1}{2m+1}$. 
    
    Intuitively, there is one validator with large stake and many validators with equal small stakes. Each service is easy to attack, but the large staked validator does not find it profitable to attack. On the other hand, each service is hard enough to attack by all small staked validators that secure it. Each service is secured by $m$ small staked validators.
    See the Figure 1 for $m=3$.

    \begin{figure}
\centering
\begin{tikzpicture}[thick,
  every node/.style={draw,circle},
  snode/.style={fill=myblue},
  vnode/.style={fill=mygreen},
  every fit/.style={ellipse,draw,inner sep=-2pt,text width=2.5cm},
  ->,shorten >= 3pt,shorten <= 3pt
]
\begin{scope}[start chain=going below,node distance=7mm]
  \node[snode,on chain] (s1) [label=left: \text{2}] {};
  \node[snode,on chain] (s2) [label=left: \text{2}] {};
  \node[snode,on chain] (s3) [label=left: \text{2}] {};
\end{scope}

\begin{scope}[xshift=4cm,yshift=-0.0cm,start chain=going below,node distance=7mm]
  \node[vnode,on chain] (v1) [label=right: \text{1/3}] {};
  \node[vnode,on chain] (v2) [label=right: \text{1/3}] {};
  \node[vnode,on chain] (v3) [label=right: \text{1/3}] {};
  \node[vnode,on chain] (v4) [label=right: \text{1/3}] {};
  \node[vnode,on chain] (v5) [label=right: \text{1/3}] {};
  \node[vnode,on chain] (v6) [label=right: \text{1/3}] {};
  \node[vnode,on chain] (v7) [label=right: \text{1/3}] {};
  \node[vnode,on chain] (v8) [label=right: \text{1/3}] {};
  \node[vnode,on chain] (v9) [label=right: \text{1/3}] {};

  \node[vnode,on chain] (v10) [label=right: 6] {};
\end{scope}

\node [myblue,fit=(s1) (s2) (s3),label=above:$S$] {};
\node [mygreen,fit=(v1) (v10),label=above:$V$] {};

\draw (v1) -- (s1);
\draw (v2) -- (s1);
\draw (v3) -- (s1);
\draw (v4) -- (s2);
\draw (v5) -- (s2);
\draw (v6) -- (s2);
\draw (v7) -- (s3);
\draw (v8) -- (s3);
\draw (v9) -- (s3);
\draw (v10) -- (s1);
\draw (v10) -- (s2);
\draw (v10) -- (s3);
\end{tikzpicture}
\label{figure}
\caption{Example construction for $m=3$.}
\end{figure}

    First, we show that $G$ is a secure restaking graph. Note that the large validator, indexed with ${m^2+1}$ and with stake $\sigma_{m^2+1}=2m$ does not want to participate in any attack, as it is losing at least the total value available across all services. No service $s\in \{1,2,...,m\}$ can be attacked by all validators securing it other than the large validator, since the service security parameter is $\alpha_s=\frac{1}{2m+1}$, while all small validators indexed $((s-1)m+1,(s-1)k+2,...,sm)$ make up stake $1$ in total. The total available stakes securing a service $s$ is $2m+1$.

To show a lower bound on $RS(G)$, first we show that when any service with value $2$ and security parameter $\frac{1}{2m+1}$ gets initial allocation of stakes in PoS protocol equal to a vector $(\frac{1}{m},...,\frac{1}{m}, 0)$, there needs to be at least an extra $m$ stake to make the PoS protocol secure. Suppose we add stakes such that new vector of stakes becomes $(a_1,...,a_m, a_{m+1})$. From now on we distinguish three cases:

\begin{enumerate}
    \item If there is a subset of validators, whose stakes sum up to a number between $1$ and $2$, then it must be that the total stake, $\sum_{i=1}^{m+1}a_i$, is at least $2m+1$. Otherwise, this subset would be able to attack the service and the attack would be profitable. This implies that the extra stake size is at least $2m=2m+1-m\frac{1}{m}$.
    \item If there is a subset of validators, whose stakes sum up to a number between $0.5$ and $1$, then it must be that the total stake, $\sum_{i=1}^{m+1}a_i$, is at least $m+1$. Otherwise, this subset would be able to attack the service and the attack would be profitable. This implies that the extra stake size is at least $m$.
    \item If there is no such subset for either of the cases 1. and 2., then it must be that more than half of the numbers among $(a_1,...,a_m,a_{m+1})$ are  at least $2$, which implies that at least $\frac{m+1}{2}(2-\frac{1}{m})\geq m$ extra stake was added to the initial stake distribution of $(\frac{1}{m},...,\frac{1}{m}, 0)$.    
\end{enumerate}

Initially, $2m$ stake of validator indexed $m^2+1$ is available, therefore, we need to add at least extra $m\cdot m-2m$ stake:

\begin{itemize}
    \item there are $m$ services to secure,
    \item for each service we need to add extra stake at least $m$.
\end{itemize}

On the other hand, original total stake amount is $T(\sigma)=3m$, therefore, restaking savings satisfy $$RS(G)\geq \frac{m\cdot m - 2m}{3m}\in \Omega(m).$$ This shows a lower bound on the required stake.

It is easy to show that $O(m^2)$ total stake is enough to add to validators in $V$, so that the resulting stake vector $\sigma\in \mathcal{W}(G)$. It can, for example, be done by adding the stake $2m^2-m$ to the validator $m^2+1$ so that its total stake becomes $(2m+1)m$. Then, assume the validator splits this total stake into $m$ equal stakes of $2m+1$, for each $m$ PoS protocols. The small stakers in these protocols can not attack any of the services even if they all coordinate, while the large staker does not find any attack profitable. This proves the claim of the theorem,  $RS(G)\in \Theta(m)$.

\end{proof}

The sets $S_L$ and $V_L$ from the proof of Theorem~\ref{non_trivial_upper_bound} end up being the sets $S$ and $V$, respectively. The highest degree in $V_L$ is equal to $m=\floor{\sqrt{n}}$. This provides an asymptotically matching lower bound to the upper bound result derived in Theorem~\ref{non_trivial_upper_bound}.

\begin{corollary}
There are instances of secure restaking graphs $G$ in which $RS(G)\in \Theta(\sqrt{n})$.    
\end{corollary}

Note that in the lower bound example construction from Theorem~\ref{lower_bound}, inverse of the security parameter $\frac{1}{\alpha_s}=2m+1$ for any $s\in S$, hence, this construction gives an asymptotically matching lower bound to the $\frac{1}{\alpha_s}-1$ upper bound as well. The upper bounds obtained in Proposition~\ref{simple_upper_bounds}, $\max_{s\in S}d_G(s)$ and $K$, are also of the order $\Theta(m)$, hence asymptotically matching the lower bound from Theorem~\ref{lower_bound}. 

The total stake in the example of Theorem~\ref{lower_bound} is equal to $3m$, while the lower bound requirement on the total stake derived from condition~\eqref{sufficient_eigenlayer} is equal to $\sum_{s\in S}\frac{\pi_s}{\alpha_s}=m(2m+1)$, which shows that in some cases the condition~\eqref{sufficient_eigenlayer} requires a substantially higher total stake than would be enough to secure the restaking graph. Moreover, the condition for any validator $v$ except the last validator $v=m^2+1$ translates into $$\frac{1}{m}\geq \frac{1/m}{2m+1} \cdot \frac{2}{\frac{1}{2m+1}}=\frac{2}{m},$$ which is violated by a factor two. The condition for the validator $v=m^2+1$
is violated by a very large multiplicative factor, as it is equivalent to: 

$$2m\geq \sum_{s=1}^{m}\frac{2m}{2m+1} \cdot \frac{2}{\frac{1}{2m+1}}=4m^2.$$

From the proof of Theorem~\ref{lower_bound} we can also observe that at least one of the  individual validator stakes should be increased by a factor of $\Theta(m)$. Therefore, if restaking savings were to be measured like over-collaterization factor of individual validators, as in~\cite{robust_restaking}, the asymptotics would be the same. This adds to the strength of the lower bound example.

\begin{definition}
    Let $RS:=\sup_{G}RS(G)$ denote the restaking savings in the extremal case on all secure restaking graphs. 
\end{definition}

Then, the direct corollary of Theorem~\ref{lower_bound} is the following:
\begin{corollary}\label{RS_infty}
    $RS =\infty$.
\end{corollary}

That is, restaking savings are not bounded by a constant, and they may grow unboundedly as the number of services and validators grows. Note that we need both the number of services and the number of validators to grow unboundedly. Otherwise, the upper bounds obtained in Proposition~\ref{simple_upper_bounds} ensure that $RS(G)$ is constant. 

\subsection{PoS Savings}

Suppose a PoS protocol centered at any service $s\in S$, having security parameter $\alpha_s$ and allocated $\sigma_v^s$ stakes for any $v\in N_G(s)$ is safe, by the same definition as above: no attacking set of validators finds it profitable to attack service $s$. 
Consider the total stake for any validator $v\in V$ defined as $\sigma_v = \sum_{s\in N_G(v)}\sigma_v^{s}$. That is, we aggregate all stakes that the validator has staked across all PoS protocols in which it participates. We ask a similar question to the previous section:

\begin{question}
    How much stake do we need to add to the aggregate validator stakes, $\sigma_v$, so that the resulting restaking graph is secure? 
\end{question}

The security of the restaking graph is considered by the original definition~\ref{def:secure_restaking}. Let $U(G)$ denote the space of stake vectors for which $G$ restaking graph is secure. Then, for each restaking graph, we define: 

\begin{definition}
    For a given aggregate restaking graph $G$, PoS savings, $PoSS(G)$ denotes a minimum additional total value to stakes in $\sigma$ vector in relative terms to $T(\sigma)$ so that the resulting restaking graph $G$ is secure:
    \begin{equation*}
        PoSS(G)=\min_{\sigma'\in D(\sigma)\cap U(G)}\frac{T(\sigma')-T(\sigma)}{T(\sigma)}
    \end{equation*}
\end{definition}

$PoSS(G)$ thus captures the inefficiency of aggregation, where interdependencies may necessitate over-collateralization to restore security.
First, we note the following inequality that relates the total value to capture from all services and the total staked amount of validators. It will be useful to obtain an upper bound on the PoS savings.

\begin{lemma}\label{lemma:inequality}
    The total stake amount of all validators is more or equal than the total value derived from all projects:  $T(\sigma)\geq T(\pi)$.
\end{lemma}

\begin{proof}
    A PoS protocol centered at $s\in S$ is secure implies that 
    
    \begin{equation}\label{PoS_condition}
        \sum_{v\in N_G(s)}\sigma_v^s\geq \pi_s,
    \end{equation} as otherwise all validators in $N_G(s)$ would be able to profitably attack the service $s$. Summing up left hand side of~\eqref{PoS_condition} for all $s\in S$ and the definition of $\sigma_v$ imply the claim of the lemma:

    \begin{equation*}
        T(\sigma)=\sum_{v\in V}\sigma_v = \sum_{s\in S}\sum_{v\in N_G(s)}\sigma_v^s \geq \sum_{s\in S}\pi_s = T(\pi).
    \end{equation*}
\end{proof}

Similarly to restaking savings $RS$, we can define a measure $PoSS$, which is maximized on all instances of the underlying graph $G$ and initial secure PoS protocols.

\begin{definition}
Let $PoSS:=\sup_{G,\sigma_v^S}PoSS(G)$ denote the extreme value of PoS savings, where the stake vector $\{\sigma_v^s: v\in N_G(s)\}$ constitutes the secure PoS protocol for any $s\in S$.
\end{definition}

We obtain the following lower bound on the $PoSS$.

\begin{proposition}\label{prop:lower_bound_PoSS}
    There are instances of $G$ so that    $PoSS(G)\geq m - 1.$
\end{proposition}

\begin{proof}
    Suppose there are $m$ services, each with a value of $\pi_s=1$. Suppose there are $m$ validators, and assume each validator secures all $m$ services. 
For each service $s$, there is its corresponding dedicated validator $s$ with a stake of $\sigma_s^s=1$, and the other $m - 1$ validators a stake of $\epsilon$ each. 
Set the security parameter for each service to $\alpha_s = \frac{(m-1)\epsilon}{1 + (m-1)\epsilon}$. First, we check the security of these PoS protocols. An attacking subset containing the dedicated validator would lose at least 1, which results in no profit. A subset excluding the dedicated validator has a maximum combined stake of $(m-1)\epsilon$, which gives a stake fraction of exactly $\alpha_s$. 
Because an attack requires a stake fraction strictly greater than $\alpha_s$, this subset cannot attack. Therefore, the independent PoS protocols are secure.
The initial total stake is $T(\sigma) = m(1 + (m-1)\epsilon)$.

When aggregated into a restaking graph, each validator's total stake becomes $1 + (m-1)\epsilon$. 
Because all validators secure all services, the total stake on any given service is $m(1 + (m-1)\epsilon)$.
Consequently, the fraction of the total stake controlled by a single validator on any service is exactly $\frac{1}{m}$. If we choose $\epsilon$ to be sufficiently small (specifically $\epsilon < \frac{1}{(m-1)^2}$), then $\frac{1}{m}$ becomes strictly greater than $\alpha_s$. This means any single validator has enough stake fraction to independently attack any service. 
Since this applies to all services, every validator can profitably attack all $m$ services, gaining a value of $m$ while only losing a stake of roughly 1.

To secure this aggregated graph, we must eliminate this profitable attack for all $m$ individual validators. For any given validator, we can remove the attack in one of two ways:

\begin{itemize}
\item Increase the total stake on enough services to force the validator's stake fraction to drop below $\alpha_s$.
\item Increase the validator's total stake to $\ge m$, so that attacking m services is no longer profitable.
\end{itemize}

Because $\epsilon$ can be chosen to be arbitrarily close to $0$, blocking the attack via the $\alpha_s$ threshold requires an arbitrarily large amount of extra stake (inversely proportional to $\epsilon$) to inflate the denominator of the attacking fraction. 
Therefore, the strictly cheaper option is to increase every validator's stake to $m$.
The extra stake required for each validator is $m - (1 + (m-1)\epsilon)$. Across all $m$ validators, the total extra stake $E$ is $m(m - 1 - (m-1)\epsilon)$.

Then, PoSS for this instance is the ratio of the extra stake to the initial total stake:
$${PoSS}(G) = \frac{E}{T(\sigma)} \ge \frac{m(m - 1 - (m-1)\epsilon)}{m(1 + (m-1)\epsilon)}.$$

Taking the limit as $\epsilon$ approaches 0, this ratio simplifies exactly to $m - 1$.
\end{proof}

The lower bound can also be interpreted as $n-1$, since by construction the number of validators is equal to the number of services. This stands in contrast to the upper bound obtained for restaking savings, which was of order $\Theta(\sqrt{n})$. 
Similarly to Corollary~\ref{RS_infty}, we obtain:

\begin{corollary}\label{POSS_infty}
$PoSS=\infty.$
\end{corollary}

Next, we provide an upper bound on $PoSS(G)$, similar to the upper bound obtained for restaking savings in Proposition~\ref{simple_upper_bounds}.

\begin{proposition}\label{prop:upper_bound_PoSS}
    $PoSS(G)\leq \max_{s\in S}d_G(s)-1$. 
\end{proposition}

\begin{proof}
    First, we show that $PoSS(G)\leq \max_{s\in S}d_G(s)$, for the intuition. For each service $s\in S$, add to all its validators $v\in N_G(s)$ a stake equal to $\pi_s$. The resulting restaking graph is secure. For any subset of validators, $U\subseteq V$, they lose more cumulative stake than the value they derive from the maximal set they attack, $M(U)$, resulting into a not profitable attack. The reason for this inequality is that by our construction, for any $s\in M(U)$, at least one of the validators in $U$ received an additional stake of $\pi_s$. We used an extra stake of $$\sum_{s\in S}d_G(s)\pi_s\leq \max_{s\in S}d_G(s)T(\pi)\leq \max_{s\in S}d_G(s)T(\sigma),$$ where the last inequality is derived in Lemma~\ref{lemma:inequality}. This completes the proof of the claim. Note that in this proof, we did not use original stakes validators are endowed to. Now suppose that to each validator $v\in V$ we add a stake equal to $\max(0, \sum_{s\in N(v)}\pi_s-\sigma_v)$. It is clear from the proof of the first claim that the resulting graph is secure. 
    
    In the following, we will try to upper bound the value $$E:=\sum_{v\in V}(\max(0, \sum_{s\in N_G(v)}\pi_s-\sigma_v)).$$ Let $\Delta_v:=\sum_{s\in N_G(v)}\pi_s-\sigma_v$. The set of validators can be divided into two disjoint subsets $V^+$ and $V^-$, so that $V^+:=\{v\in V: \Delta_v\geq 0\}$ and $V^-:=\{v\in V: \Delta_v < 0\}$. Let $S^+$ denote the set $\{s\in S: N_G(s)\subseteq V^+\}$ of the services whose validators are all in $V^+$, and $S'=S\setminus S^+$ be the set of remaining services, each having at least one validator in $V^-$. For $s\in S'$, $|N_G(s)\cap V^+|\leq d_G(s)-1$. Then, 
    \begin{equation}
        E=\sum_{v\in V^+}\Delta_v = \sum_{v\in V^+}(\sum_{s\in N_G(v)}\pi_s - \sigma_v) = \sum_{s\in S}\pi_s|N_G(s)\cap V^+| - \sum_{v \in V^+}\sigma_v.
    \end{equation}

    Next, the following chain holds: 

    \begin{align*}
        &\sum_{s\in S}\pi_s|N_G(s)\cap V^+| = \sum_{s\in S^+}\pi_s|N_G(s)\cap V^+| + \sum_{s\in S'}\pi_s |N_G(s)\cap V^+| \leq \\
        & \sum_{s\in S+}\pi_sd_G(s) + \sum_{s\in S'}\pi_s(d_G(s)-1).
    \end{align*}

    Therefore, the extra added value to stakes can be upper bounded as $$E\leq \sum_{s\in S+}\pi_s d_G(s) + \sum_{s\in S'}\pi_s (d_G(s)-1) - \sum_{v\in V^+}\sigma_v.$$ 

    Since each PoS protocol is secure, we have $\sum_{v\in N_G(s)}\sigma_v^s\geq \pi_s$ for all $s\in S$, which implies that:

    \begin{equation*}
        \sum_{s\in S^+}\sum_{v\in N_G(s)}\sigma^s_v\geq \sum_{s\in S^+}\pi_s.
    \end{equation*}

    For $s\in S^+$, $N_G(s)\subseteq V^+$, so the left hand side of the above is: 

    \begin{equation}
        \sum_{v\in V^+} \sum_{s\in S^+\cap N_G(v)}\sigma_v^s \leq \sum_{v\in V^+} \sum_{s\in N_G(v)}\sigma_v^s=\sum_{v\in V^+}\sigma_v.
    \end{equation}
    Therefore, $\sum_{v\in V^+}\sigma_v\geq \sum_{s\in S^+}\pi_s$. 
    Substituting in the upper bound on $E$, we get: 

    \begin{equation*}
        E\leq \sum_{s\in S^+}\pi_sd_G(s)+\sum_{s\in S'}\pi_s(d_G(s)-1)-\sum_{s\in S^+}\pi_s = \sum_{s\in S}\pi_s(d_G(s)-1).
    \end{equation*}

    Finally, $$\sum_{s\in S}\pi_s(d_G(s)-1)\leq (\max_{s\in S}d_G(s) - 1)T(\pi)\leq (\max_{s\in S}d_G(s) - 1)T(\sigma).$$ 
\end{proof}

Note that the upper bound in Proposition~\ref{prop:upper_bound_PoSS} exactly matches the lower bound obtained in Proposition~\ref{prop:lower_bound_PoSS}, unlike the case of restaking savings, in which case the lower and upper bounds are only asymptotically matching.

\subsection*{Interpretation of Results}

Although mostly conceptual, our results can provide some design recommendations. 
For an existing restaking project that has an underlying graph structure already with specified stake sizes and service valuations, we could compute restaking savings. 
For a given restaking graph, it must be feasible to check if it is secure using heuristic algorithms, instead of checking a polynomial-time sufficient condition which severely limits the domain and may result into zero restaking savings.

The upper bound obtained in Proposition~\ref{RS_upper_bound_alpha} is particularly relevant in practical cases, since the security parameters of most protocols are usually from the set of fractions $\{\frac{1}{4}, \frac{1}{3},\frac{1}{2}\}$. These values come from Byzantine fault-tolerant mechanisms, used to achieve consensus on the state of the chain. In such cases, the bound simplifies to at most $2$ to $4$ times the original total stake, providing a tight estimate of the capital overhead required to achieve equivalent PoS security.

The number of unique addressed validators is high for the large-scale ecosystems such as Ethereum chain, approaching one million, while the real number of validators can still be on the order of tens of thousands. All validators are staked for Ethereum, that is, the restaking protocol based on Ethereum has the highest degree $\max_s{d_G(s)} = n$. The number of projects that require staking services is typically low, in the order of hundreds, and it is highly unlikely that a validator would provide services for most of them, therefore, it would be expected that $\max_v(d_G(v))$ remains low.

Recently, EigenLayer's restaking project upgraded the system that only allows a unique slashing property: any staked asset can be slashed by a single service~\footnote{https://blog.eigencloud.xyz/introducing-the-eigenlayer-security-model/}. This restriction effectively makes it equivalent to validators partition their total stake across services, and hence results into separate secure PoS protocols.
As a result, the system no longer realizes capital efficiency gains from true stake reuse, and the restaking savings collapses to zero under this constraint.

\section{Conclusion}
We introduce a formal framework for comparing the capital efficiency of secure restaking protocols and equivalent independent Proof-of-Stake (PoS) protocols, measured by the total stake.
Our main results establish tight asymptotic bounds on the restaking savings (the relative extra capital needed to redistribute stakes from a secure restaking graph into separate secure PoS protocols): we prove matching upper and lower bounds of $\Theta(\sqrt{n})$, where $n$ is the number of validators. These bounds are achieved through greedy covering technique and explicit construction, showing that restaking can deliver meaningful capital efficiency gains in large-scale systems.

In the complementary direction, we demonstrate that aggregating secure PoS protocols into a secure restaking graph can incur a significant overhead. We provide matching lower and upper bounds on PoS savings (the extra stake needed post-aggregation to restore security), equal to maximum service degree in the restaking graph minus one. 

Taken together, our results indicate that neither restaking nor independent PoS protocols dominate the other in terms of capital requirements -- each has structural advantages depending on the underlying graph.

One important open question remains. 
It is designing a generic, polynomial-time-checkable sufficient condition for restaking graph security that admits graphs with strictly positive restaking savings.

\bibliographystyle{plain}
\bibliography{sample}

\end{document}